\documentclass[]{article} 

\usepackage{balance} 
\usepackage{booktabs}
\usepackage{nicefrac}
\usepackage{amsfonts}       
\usepackage{amsmath}
\usepackage{amsthm}
\usepackage{hyperref}
\usepackage{tikz}
\usepackage{xcolor}
\usepackage{comment}
\usepackage{bbm}
\usepackage{natbib}
\usepackage{thmtools}  
\usepackage{thm-restate}
\usepackage{subfig}
\usepackage{todonotes}

\usepackage{pgfplots}
\DeclareRobustCommand{\citet}[1]{\citeauthor{#1}~\citeyear{#1}}

\title{Measuring and Controlling Divisiveness in Rank Aggregation}

\author{
Rachael Colley$^1$
\and
Umberto Grandi$^1$\and
C\'esar Hidalgo$^{2,3,4}$\and
Mariana Macedo$^2$ \and
Carlos Navarrete$^2$
}

\date{\{rachael.colley,umberto.grandi\}@irit.fr,
\{cesar.hidalgo, mariana.macedo, carlos.navarrete\}@univ-toulouse.fr\\
$^1$IRIT, Universit\'e Toulouse Capitole, France\\
$^2$Center for Collective Learning, ANITI, TSE, IAST, IRIT, Université de Toulouse, France\\
$^3$Alliance Manchester Business School, University of Manchester, UK\\
$^4$Center for Collective Learning, CIAS, Corvinus University, Hungary}

\newcommand{\BibTeX}{\rm B\kern-.05em{\sc i\kern-.025em b}\kern-.08em\TeX}

\newtheorem{example}{Example}

\newtheorem{definition}{Definition}

\newtheorem{proposition}{Proposition}

\newcommand{\N}{\mathcal{N}} 
\newcommand{\F}{\mathcal{F}} 

\renewcommand{\P}{\mathcal{P}} 
 
\newcommand{\I}{\mathcal{I}}

\newcommand{\Div}{\textsc{Div}}

\newcommand{\Rank}[2]{\mathit{rank}(#1,#2)} 
\newcommand{\Var}{Var}

\newcommand{\score}{s}
\newcommand{\X}{\mathcal X}
\newcommand{\borda}{\mathit{Borda}}
\newcommand{\copeland}{\mathit{Cop}}

\newcommand{\inject}[1]{\textsc{Inject}_{#1}}

\begin{document}




\maketitle

\begin{abstract}
In rank aggregation, members of a population rank issues to decide which are collectively preferred. 
We focus instead on identifying divisive issues that express disagreements among the preferences of individuals.   
We analyse the properties of our divisiveness measures and their relation to existing notions of polarisation. 
We also study their robustness under incomplete preferences and algorithms for control and manipulation of divisiveness.   
Our results advance our understanding of how to quantify disagreements in collective decision-making.
\end{abstract}


\section{Introduction}

Rank aggregation is the problem of ordering a set of issues according to a set of individual rankings given as input. 
This problem has been studied extensively in computational social choice (see, e.g., \citet{brandt2016handbook}) when the rankings are assumed to represent human preferences over, for example, candidates in a political election, projects to be funded, or more generally alternative proposals. 
The most common approach in this literature is to find normative desiderata for the aggregation process, including computational requirements such as the existence of tractable algorithms for its calculation and characterisations of the aggregators that satisfy them.
Rank aggregation also has a wide spectrum of applications from metasearch engines \citep{dwork2001rank} to bioinformatics \citep{kolde2012robust}, receiving attention also in the statistical ML literature \citep{korba2017learning}.

Previous work on rank aggregation has focused on how to best elicit which issues are the most agreed upon, without identifying the issues that divide them. 
Instead, a wide literature in economics and the social sciences has developed measures of social, economic, and political polarisation. 
Classical work analysed polarisation in the distribution of wealth, goods, and opinions \citep{esteban1994,duclos2004polarization}, showing that well-studied notions of inequality are unfit to measure polarisation as they do not consider the weight of sub-populations.  
Another common approach in social science uses the variance of distributions to measure polarisation (see, e.g.,  \citet{musco2018minimizing,gaitonde2020adversarial1}). 

%


In this paper, we put forward a family of functions that starting from a collection of individual rankings are able to order issues based on their divisiveness.
We compute an issue's divisiveness by aggregating the disagreement among all possible sub-populations defined by the relative preference among the other issues.
Our proposal relates to the literature on three important aspects.
First, a parameter in our definition allows us to move from an adaptation of the classical measure of polarisation from \citet{esteban1994} at one end of the spectrum to the detection of disagreements from minorities at the other end.
Second, our measures are parameterised by the rank aggregation function that is used to compute the most agreed-upon issues.
In this way, we can align our notions of divisiveness with the functions chosen to measure the agreement of the population.\footnote{This is particularly important in social choice applications, where the individual preferences collected are a (possibly strategic) response to the collective decision-making rule chosen.} 
Third, while existing work focused on comparing different sets of rankings based on polarisation \citep{can2015measuring}, diversity \citep{HashemiE14}, or cohesiveness \citep{alcantud2015pairwise}, here, we aim at identifying the most divisive issues \emph{within} a complete profile of rankings. 
In doing so, we do not need to assume that issues are independent, as common in the social choice literature.



Our work can further guide how to query a population towards being more inclusive and unified, e.g., through deliberative instances.
This can include measures that go towards decreasing divisiveness, such as recent work suggesting the construction of recommender systems to depolarise a population \citep{Stray22FirstMonday}, or simply take advantage of this information when steering the public debate (recent work from \cite{ash2017elections} suggests that politicians spend more time on divisive topics than on neutral ones). 
Our work also contributes to social choice theory, where related notions of preference diversity have shown to have effects on the probability of paradoxes \citep{gehrlein2010voting}, the competitiveness in matching markets \citep{halaburda2010unravelling}, or the computational complexity of manipulating an election \citep{Wu2022manipulating}. 
Moreover, our work can be useful in refining the preference analysis of applications, such as online forums or surveys, that query a population on their opinions and return aggregated information about the group as a whole.

\paragraph{Contribution and Paper Structure.}
We extend the notion of divisiveness, introduced by  \citet{navarrete2022understanding}, to a family of measures which take into account the size of a sub-population and use the well-studied scoring functions. 
These extensions allow us to connect divisiveness to measures of polarisation. 
We give a theoretical and experimental analysis of our divisiveness measures by relating them to other notions and giving bounds on their limit cases (Section~\ref{sec:measuringdiv}). Importantly, we show that our measures can distinguish between key profiles which other measures cannot. We then inspect two aspects of control, first, by studying the effect of removing pairwise comparisons from the agent's rankings (Section~\ref{sec:controladd}
) and second, by adding additional controlled agents (Section~\ref{sec:controladd}).
All our code is available at \url{https://github.com/CenterForCollectiveLearning/divisiveness-theoretical-IJCAI2023}.

\paragraph{Related Work.}

The notion of divisiveness studied in this paper builds on the work of \citet{navarrete2022understanding}, who identify the most divisive issues from proposed government programs from crowdsourced political preferences. 
Many papers start from a profile of rankings and compare them based on how \emph{consensual} (equivalently, \emph{cohesive}) they are \citep{boschcharacterizations,garcia2010consensus,alcalde2008measurement,alcalde2013measuring}, or in the opposite direction, i.e., how \emph{diverse} is the set of preferences \citep{HashemiE14,karpov2017preference}. 
In particular, \citet{alcantud2015pairwise} measures the cohesiveness of a group by aggregating the dissimilarity of their orderings (similarly to \citet{can2015measuring} who however focus on polarisation).
Most of these settings are based on pairwise comparisons, except for \citet{alcalde2016we} and \citet{xue2020quantifying}, who look at patterns of varying sizes in the rank profile.
One of the methods proposed by \citet{HashemiE14} to measure preference diversity is to average the distance between the individual rankings and the aggregated one. 
This is in line with our approach, but it only provides a measure to compare different populations of rankings without going to the level of single issues.
We note that also an influential theory of diversity not based on preferences but on (binary) features of not necessarily independent alternatives has been proposed by \citet{nehring2002theory}.

\section{Basic Definitions}\label{sec:basicdef}

This section introduces the basics of rank aggregation and scoring rules. We put forward our definition of divisiveness, then we compare it with existing notions of polarisation.

\subsection{Preliminaries}

\paragraph{Rankings.}
Let $\I=\{a,b\dots\}$ be a finite non-empty set of $m$ issues.
A strict ranking (aka. linear order) on $\I$ is an asymmetric, transitive, and complete binary relation on $\I$. We let  $a\succ b$ denote the fact that $a$ is strictly preferred to $b$ in the ranking $\succ$.
In what follows, we will write $\succ=acdb$ for $\succ=a\succ c \succ d\succ b$, reading preferences from left to right.
The set of all strict rankings over $\I$ will be denoted by $\mathcal L(\I)$. We denote with $\Rank{a}{\succ}$ the rank of $a$ in $\succ$ with the first position being $1$ and the last being $m$.

\paragraph{Individual Rankings.}
Let a finite non-empty set of $\N=\{1,\dots, n\}$ agents express a strict rankings over $\I$ (sometimes referred to as preferences).
We let $\P=(\succ_1, \dots, \succ_n)$ denote the resulting profile of rankings, where $\succ_i$ is agent's $i$ ranking over $\I$. 
We let $\N_{a\succ b}=\{i\in \N \mid a \succ_i b\}$ be the set of voters in $\N$ who prefer $a$ to $b$. 
The restriction of profile $\P$ to the agents in $\mathcal X$ is denoted by $\P_\mathcal X=\langle \succ_i\mid i\in \mathcal X \rangle$. 
When $\mathcal X =\N_{a\succ b}$ we simply write $\P_{a\succ b}$.
We call $\P$ a consensual profile if for all $i,j\in\N$ we have that $\succ_i=\succ_j$.
Every preference profile $\P$ can be represented as a weighted (anonymous) profile, i.e., as a set of pairs $(w_j,\succ_j)$ indicating that $w_j\in \mathbb N$ agents have preference $\succ_j$.

\paragraph{Collective Scoring of Issues.}

Rank aggregation functions define a collective ranking of issues based on the agreements among the individual rankings in a profile. 
A large number of rules have been proposed and analysed in the literature on (computational) social choice and artificial intelligence.
We focus on rank aggregators defined by a scoring function, where the collective ranking over issues is obtained via a function  $\score: \I \times \mathcal L(\I)\to [0,1]$ that assigns a score to each issue in a given profile.
Notable examples are the (normalised) Borda score, which counts the number of issues strictly preferred to a given issue, $\borda(a,\P)=\sum_{b\in \I\backslash \{a\}} \frac{\#(\N_{a>b})}{n\cdot (m-1)}$, where $\#(\X)$ is the cardinality of a set $\X$.
Or the normalised Copeland score, which counts the number of majority contests won by an issue, $\copeland(a,\P)=\frac{\#\{b\in \I\backslash \{a\}
\,\mid\, \#(\N_{a>b})>\#(\N_{b>a})\}}{m-1}$.

\subsection{Divisiveness}

For a given sub-population $\X\subseteq \N$ and issue $a\in \I$, we measure the divisiveness of $a$ for $\X$ as the difference between the collective scoring of $a$ in sub-population $\X$ and in its complement sub-population $\N\backslash \X$. 

\begin{definition}{\citep{navarrete2022understanding}}
The divisiveness of an issue $a\in\I$ with respect to a sub-population $\X\subseteq\N$ in profile $\P$ is defined as:
$$\Div^\score(a,\X,\P) = |\score(a,\P_{\X}) - \score(a,\P_{\N\backslash \X})|.$$
If $\X=\emptyset$ or $\X=\N$, we set $\Div^\F(a,\X,\P)=0$.
\end{definition}

Examples of a sub-population $\X$ can be descriptive, such as agents living in cities (thus $\N\backslash\X$ are agents living in rural areas) or agents with a given political orientation (thus, $\N\backslash\X$ would be those who do not ascribe to this orientation).
We can now give a definition of divisiveness for issue $a$ that is independent of a given sub-population by averaging over all sub-populations $\N_{a\succ b}$ for all other issues $b$.  
We include an additional parameter $\alpha$ that allows us to take into consideration the size of a sub-population allowing for a weighted average version of an issue's divisiveness.

\begin{definition}\label{def:div}
The $\alpha$-divisiveness $\Div_\alpha^\score(a,\P)$ of an issue $a{\in}\I$ in profile $\P$, with $\alpha\in[0,1]$ and $\ell\in\mathbb R^+$, is defined as:
$$\frac{1}{m-1}\sum\limits_{\substack{b\in\I \\ a\not=b}} \left( \ell\frac{\#(\N_{a>b})\cdot\#(\N_{b>a})}{n^2}\right) ^\alpha \Div^\score(a,\N_{a\succ b},\P).$$
\end{definition}
When $\alpha=0$ and $\score=\borda$, our definition is a reformulation of the divisiveness measure from \citet{navarrete2022understanding}. 
When $\alpha=1$, Definition~\ref{def:div} can be interpreted as one of the polarisation measures of \citet{duclos2004polarization} and \citet{esteban1994}, calculated on the distribution of ranks that issue $a$ received from individuals. 
We refer to the multiplicative factor of the measure defined in Definition~\ref{def:div} in each summand as the $\alpha$-factor. 
The $\alpha$-factor is maximal when $\#(\N_{a\succ b })=\#(\N_{b\succ a})=\nicefrac{n}{2}$, and we will often set $\ell=4$ so that the $\alpha$-factor is at most $1$.  $\ell$ is a normalising factor left open in line with \citet{esteban1994}.

Intuitively, as $\alpha$ increases in Definition~\ref{def:div}, the relevance of the size of the disagreeing sub-population increases. 
The following example presents the limit case of $\alpha=0$ where the size of disagreeing sub-populations is ignored, resulting in a different divisiveness ranking than $\alpha=1$.
 
\begin{example}\label{ex:begin}
Consider $2k$ agents giving their preferences in profile $\P$ over issues $\I=\{a,b,c,d,e,f\}$ as such:
 \begin{center}
 \begin{tabular}{l|l}
$\succ_1$ ($1$ agent) &  $c a d f e  b$\\
$\succ_2$ ($k-1$ agents) & $b a d f e c$ \\
$\succ_3$ ($k$ agents) & $b  e d f a  c$
\end{tabular}  
\end{center}

Assume now that $k=5$ and $\ell=4$ (the normalising factor). Table~\ref{tab:divworkingoutex1} presents, for issue $a$, the table of disagreements (in terms of Borda score difference) between the 5 possible sub-populations defined by the pairwise comparisons of $a$ with the remaining issues.

\begin{table}[h]
    \centering
    \begin{tabular}{ccccc}
        Issue $x$ & $s(\P_{a\succ x})$ &$s(\P_{x\succ a})$ & disagreement & $\alpha$-factor\\
        \midrule
        $b$ & $0.8$ &  $0.4\overline{6}$ &  $0.\overline{3}$ & $0.36$\\
        $c$ & $0.4\overline{6}$ &  $0.8$  & $0.\overline{3}$ & $0.36$\\
        $d$ & $0.8$ & $0.2$ & $0.6$& $1$\\
        $e$ & $0.8$ & $0.2$ & $0.6$ & $1$\\
        $f$ & $0.8$ &  $0.2$ & $0.6$ & $1$\\
    \end{tabular}
    \caption{Details of the calculations for divisiveness for issue $a$, setting $\ell=4$ and $k=5$ and with $s=\borda$.}
    \label{tab:divworkingoutex1}
\end{table}

From Table~\ref{tab:divworkingoutex1}, we can compute $\Div_0^\borda(a, \P)$ by averaging the disagreements and weighting by the $\alpha$-factors to obtain $\Div_1^\borda(a, \P)$. Repeating this process for every issue gives us the values of divisiveness in Table~\ref{tab:beginexdivmeasures}.

\begin{table}[h]
    \centering
    \begin{tabular}{cccc}
    $x$ &$\borda(x,\P)$& $\Div_0^\borda(x, \P)$     & $\Div_1^\borda(x, \P)$    \\
    \midrule
      $a$   &$0.5$ & $0.49\overline{3}$ & $0.408$\\
      $b$& $0.9$  & $1$ & $0.36$ \\
      $c$& $0.1$  & $1$ & $0.36$ \\
      $d$  & $0.6$ &$0$&$0$\\
      $e$   & $0.5$&$0.49\overline{3}$&$0.408$\\
      $f$  & $0.4$&$0$&$0$\\
      
    \end{tabular}
    \caption{Normalised Borda score and divisiveness for $\alpha=0,1$ of issues in the profile of Example~\ref{ex:begin} with $k=5$ and $\ell=4$.}
    \label{tab:beginexdivmeasures}
\end{table}

The most divisive issue for $\Div_0^\borda$ is $b$ and $c$, who are at the opposite extreme of the ordering for the first agent with respect to the rest of the population, while the most divisive for $\Div_1^\borda$ are $a$ and $e$, who are second or second-to-last for all agents, yet this divides the population into two. Observe that this holds for any $k\geq 5$, showing a class of examples where the ranking of $\alpha$-divisiveness differs significantly depending on if $\alpha=0$ or $\alpha=1$.
\end{example}

\subsection{Rank-variance}\label{sec:variance}

Related literature in the social sciences often measures polarisation using notions of variance, which we now adapt to rankings. Let $\mu_\P(a)= \frac{1}{n}\sum_{i\in \N}\Rank{a}{\succ_i}$ be the average rank of an issue $a\in \I$ in a profile of rankings $\P$. The \emph{rank-variance} of issue $a$ is defined as follows:

$$\Var(a, \P)= \frac{1}{n}\sum\limits_{ i\in \N } (\Rank{a}{\succ_i}-\mu_\P(a))^2$$

The following example shows a profile in which the ranking of issues by variance differs from divisiveness (assuming $\alpha=0$ and $s=\borda$).

 \begin{example}
 Consider the following preference profile: 
  \begin{center}
 \begin{tabular}{l|l}
$\succ_1$ (10 agents) &  $a b c d e$\\
$\succ_2$ (10 agents) & $ebcda$ \\
$\succ_3$ (1 agent) & $acbde$\\
$\succ_4$ (1 agent) & $ebdca$ 
\end{tabular}  
\end{center}

The rank-variance and the divisiveness using Borda and $\alpha=0$ for each issue is the following:

\begin{center}
\begin{tabular}{cccccc}
                & a & b & c & d & e\\
\midrule
 $\Var(x,\P)$    & $4$& $0.0\overline{45}$ &$0.0\overline{90}$ & $0.0\overline{45} $& $4$\\
 $\Div^\borda(x)$  &$1$ &$ 0.074$ & $0.037$ & $0.074$ & $1$
\end{tabular}
\end{center}

This shows that issue $c$ has higher variance than issues $b$ and $d$ but lower divisiveness (Kendall's tau correlation between the two rankings is $\tau\approx0.5$).
\end{example}

Unlike other notions of polarisation, we focus on the comparisons between our divisiveness measures and rank-variance as they both return information about a single issue rather than about the population as a whole.

\section{Measuring Divisiveness and  Polarisation}\label{sec:measuringdiv}

In this section, we present some basic properties of our definition of divisiveness, showing in particular that with $\alpha=0$ it does not coincide nor is correlated with standard notions of polarisation. We also show how our definitions can be used to identify the sub-population that is most divided on an issue.

\subsection{Divisiveness Bounds}

We first observe that if $\score$ is a polynomially computable function, then so is $\Div^\score_\alpha$ for any $\alpha$. Moreover, if $\score$ is anonymous and neutral (as in the classical social choice terminology\footnote{A function taking profiles of linear orders as input is anonymous if permuting the individual rankings does not change the result. It is neutral if all issues are treated equally, i.e., permuting the name of the issues results in the ranking obtained by applying the same name permutation to the previous result.}) so is $\Div^\score_\alpha$ for any $\alpha$. 
A direct consequence of our definitions is that $0 \leq \Div^\score_\alpha(a,\P)\leq 1$, for any $\alpha$.  We now characterise the extremes of the spectrum.

First we give the sufficient conditions for minimal divisiveness with the Borda and Copeland scorings.
Let a profile $\P$ be rank-unanimous on $a$ if for all $i,j\in \N $ we have that $\Rank{a}{\succ_i}{=}\Rank{a}{\succ_j}$. Profile $\P$ is unanimous if $\succ_i=\succ_j$ for all $i,j\in\N$.
The next result shows that divisiveness is minimal when consensus is maximal, following this, we will discuss the converse of this statement in a few different ways. 

\begin{proposition}\label{prop:bounds}
If profile $\P$ is rank-unanimous on $a$ then $\Div_\alpha^\borda(a,\P)=0$, while not necessarily true for $\Div_\alpha^\copeland$.
If $\P$ is unanimous, then $\Div_\alpha^\copeland(a,\P)=0$ for all $a\in\I$. 
\end{proposition}

\begin{proof}
If $\P$ is rank-unanimous on $a$, then for all $b\in \I\backslash\{a\}$ we have that $\borda(a,\P_{a\succ b})-\borda(a,\P_{b\succ a})=0$, and thus  $\Div^\score_\alpha(a,\P)=0$.
This does not necessarily hold for $\Div^\copeland$: consider 3 agents with the following rankings over issues $\I=\{a,b,c,d\}$, forming profile $\P$: $\succ_1= abcd$, $\succ_2 = cbad$, and $\succ_3 = dbac$. Observe that for all $i\in \N$ that $\Rank{b}{\succ_i}=2$, hence $\P$ is rank-unanimous on $b$. However, $\Div^\copeland_0(b,\P)=\frac{1}{3}$. 
Finally, if $\P$ is unanimous, any notion of divisiveness will be equal to zero, as either $\N_{a\succ b}$ or $\N_{b\succ a}$ is empty for any pair $a,b\in\I$.
\end{proof}


Profile $\P$ is \emph{fully polarised} if it is split into two equally-sized sub-populations with completely opposite preferences ($m$ is even). If $n$ is even, $\P$ is fully polarised if $\succ_i=\succ_1$ for $i=1,\dots,\nicefrac{n}{2}$, and $\succ_i=\succ_2$ for $i=\nicefrac{n}{2}+1,\dots,n$, where $\Rank{a}{\succ_1}=m-\Rank{a}{\succ_{2}}$ for all $a\in \I$. If $n$ is odd, one of the sub-populations has one more agent than the other. 
Let $a$ be the top-ranked issue in $\succ_1$ (hence ranked last in $\succ_2$).
We show that divisiveness is maximal in such profiles:

\begin{proposition}\label{prop:polarised=1}
If $\P$ is fully polarised and $n$ is even, then $\Div_\alpha^\borda(a,\P)=1$ when $\ell=4$, where $a$ is the top-ranked issue in one of the two sub-populations.
If $\P$ is fully polarised and $n$ is odd, then $\Div_\alpha^\copeland(a,\P)$ equals the $\alpha$-factor.
\end{proposition}
\begin{proof}
If $\P$ is fully polarised and $n$ is even, then each summand of $\Div_\alpha^\borda(a, \P)$ is 1, as all agents  in $\N_{a\succ b}$  rank $a$ first,  all  in $\N_{b\succ a}$ rank $a$ last, and $\#(\N_{a\succ b}){=}\#(\N_{b\succ a})$. 
Thus, $\Div^\borda_\alpha(a, \P){=}1$. 
If $\P$ is fully polarised and $n$ is odd, then $a$ is a Condorcet winner in the larger sub-population and a Condorcet loser in the other. Thus, $\Div^\score(a,\N_{a>b},\P){=}1$ for each $b \in \I\backslash\{a\}$. As sub-populations differ by $1$, $\Div_\alpha^\copeland(a, \P)=\left(\frac{n^2-1}{n^2}\right)^\alpha$ when $\ell=4$. 
\end{proof}

In fully polarised profiles, two issues have maximal divisiveness: the top-ranked issues in  $\succ_1$ and $\succ_2$. Moreover, in any profile, at most two issues can have maximal divisiveness.


Finally, \emph{uniform} profiles contain each of the $m!$ possible rankings over the $m$ issues, and each ranking is equally represented in the profile (note that $n$ is even). They represent a fully noisy population of preferences.
While the measure of polarisation proposed by \citet{can2015measuring} cannot distinguish between uniform and fully polarised profiles, we show that the ranking of divisiveness is strict in the former while in the latter all issues have the same divisiveness. The following proposition is straightforward due to the symmetry of uniform profiles:

\begin{proposition}
If $\P$ is a uniform profile, then $\Div^\score_\alpha(a,\P)=\Div_\alpha^\score(a,\P)$ for all $a,b\in \I$.
\end{proposition}
For a uniform profile $\P$, $\Div_\alpha^\borda(a,\P)=\frac{1}{m-1}$ when $\ell=4$, as the average Borda score between two sub-populations $N_{a\succ b}$ and $\N_{b\succ a}$ differs by $\nicefrac{1}{m-1}$. Whereas $\Div_\alpha^\copeland(a,\P)=1$ when $\ell=4$, as the divide of the two sub-populations always ensures that $a$ always wins every majority contest in $\N_{a\succ b}$ and always loses in $\N_{b\succ a}$.

\begin{figure}[t!]
    \centering
    \includegraphics[width=0.6\columnwidth]{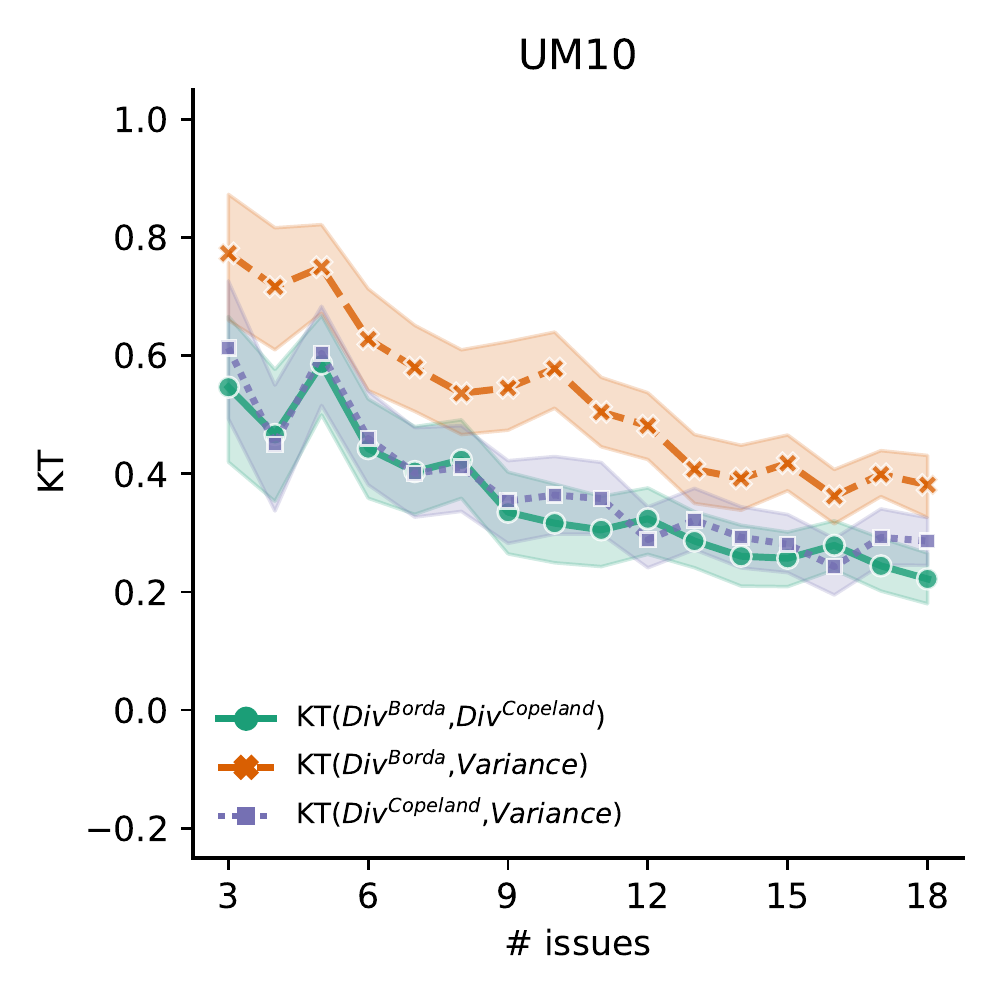}
    \caption{The average Kendall's tau correlation between each pair of divisiveness rankings using Borda and Copeland with $\alpha=0$, and the rank variance, varying the number of issues. The average is taken from $100$ profiles with $100$ agents generated by UM10.}
    \label{fig:divvar1}
\end{figure}

\subsection{Divisiveness and Rank Variance}\label{sec:divVar}

We conducted experiments on synthetic preference profiles to test whether divisiveness with $\alpha=0$ correlates with the rank-variance defined in Section~\ref{sec:variance}.
We computed Kendall's tau correlation (KT) of $\Div_0^\borda$ and $\Div_0^\copeland$ with the rank variance (cf. Section~\ref{sec:variance}). 
We tested 100 profiles of rankings generated via the impartial culture (IC) and the Urn model with a correlation of 10\% and 50\% (named UM10 and UM50, respectively). Rankings were generated using the PrefLib library \citep{mattei2017apreflib,MaWa13a}.

Using the Urn model with 10\% correlation as an example, we plot in Figure~\ref{fig:divvar1} the average Kendall's tau correlation. We observe that the three measures are correlated when profiles are on a few issues but that the values of correlation decrease significantly as we increase the number of issues, and therefore increase the possible rankings for the measures to return. The correlation between the rank variance and the divisiveness computed using Borda is higher than using Copeland. The correlation between divisiveness using Borda and divisiveness using Copeland shows a similar decreasing trend.
Similar results are captured for the impartial culture scenario, but the correlation decreases even more for the case of $KT(\Div_0^\borda,\Div_0^\copeland)$ and $KT(\Div_0^\copeland,\mathit{Variance})$.
A possible explanation of this decreasing correlation is that rank-variance only considers an issue's position in each individual ranking discarding which the rankings' structures (which are ranked above or below the issue in question). Additional details and figures can be found in the Appendix.

\subsection{Maximally Divided Sub-Populations}

In addition to providing a ranking of issues based on their divisiveness, Definition~\ref{sec:basicdef} can also be used to identify the partition that maximally divides the population for an issue. 
This is a seemingly hard computational problem, as there is an exponential number of sub-populations to consider, but we show that it can be solved efficiently for the Borda score.

\begin{proposition}
For any profile $\P$ and issue $a\in\I$, finding the sub-population $\X\subseteq \N$ that maximises $\Div_0^\borda(a,\X,\P)$ can be done in polynomial time.
\end{proposition}

\begin{proof}
Consider an arbitrary preference profile $\P$.
For an arbitrary issue $a\in \I$, we will use the following algorithm to find the sub-population $\X$ such that $\Div_0^\borda (a, \X, \P)$ is maximal. 
We first order the agents with respect to their ranking of issue $a$. Thus, without loss of generality, we can assume that for each $i\in [1,n-1]$ we have that ${\Rank{a}{\succ_{i}}} \geq \Rank{a}{\succ_{i+1}}$. 
Note that if two agents rank $a$ at the same level, their ordering is irrelevant.
We will now prove that any sub-population $\X$ that gives the maximum value of $\Div^\borda (a, \X, \P)$ will partition $\N$ such that for some $k\in [1,n-1]$ we have that $\X=\{1, \cdots k\}$ and thus $\N\backslash \X=\{k+1, \cdots, n\}$. 
Our polynomial algorithm then tests each of the $n-1$ partitions defined by $\X_k=\{1, \cdots, k\}$ and returns the one that maximises $\Div_0^\borda (a, \X_k, \P)$. 
Clearly, in each of these partitions $\X_k$, calculating $\Div_0^\borda (a, \X_k, \P)$ can be done in polynomial time.

We now prove that any $\X$ that maximises $\Div_0^\borda (a, \X, \P)$ is of the form $\X_k=\{1, \cdots, k\}$ for some $k$.
To do so, we consider some $\X$ such that there exists some $i\in \X$ yet $i-z\notin\X$, for $z\in [1, i-1]$. We need to show that the set $\X'=(\X\backslash\{i\})\cup\{i-z\}$  is more divisive than $\X$, thus  $\Div_0^\borda (a, \X', \P)\geq \Div_0^\borda (a, \X, \P)$.  
It is clear that $\borda(a, \P_{\X})\leq \borda(a,\P_{\X'})$ due to our assumption on the ordering of the agents in $\P$ in decreasing order of rank of $a$.  
For similar reasons, we have that $\borda(a, \P_{\N\backslash\X})\geq \borda(a,\P_{\N\backslash\X'})$. 
Therefore, it must be the case that $\Div_0^\borda (a, \X', \P) \geq\Div_0^\borda (a, \X, \P)$. 
Thus, we can transform any $\X$ that maximises divisiveness into a sub-population of the form $\X_k$ by performing a finite number of swaps, as described above.
\end{proof}

Determining the complexity of finding maximally divided sub-populations under $\Div_0^\copeland$ does not seem trivial, and we leave it as an open problem.

\section{Divisiveness Control}\label{sec:manip}

Manipulation and control have been widely studied for rank aggregation procedures. Manipulation is when an individual misrepresents their reported ranking to improve the score of their favourite candidate. A classical result showed that manipulation can be performed in polynomial time for the Borda score \citep{bartholdi1989computational}. Instead, control is when an external agent aims at altering the score of a designated candidate by performing actions such as removing agents or candidates from the profile of ranking. To give an example, preventing a candidate from being the Copeland winner by adding new rankings was shown to be a computationally hard problem.
For an introduction to the computational complexity of these problems, we refer to the survey from \citet{faliszewski2016control}. 
In this section, we focus on two approaches to control the measure of divisiveness: (i) by removing pairwise comparisons from the agents' rankings and (ii) by adding new agents (which could, e.g., be performed by \emph{bots} on any platform that crowdsources individual preferences). 
This section focuses on the less studied divisiveness measure with $\alpha{=}0$ and hence we omit $\alpha$ from $\Div^s_\alpha$.


\subsection{Removing Pairwise Comparisons}\label{sec:robustness}

This section studies the disruption of the divisiveness measure by the deletion of pairwise comparisons from the rankings. 
This can be thought of as sabotage, as the control actions here do not have a clear goal, e.g., making a single issue the most divisive one. Instead, the aim of this control problem is to disrupt the divisiveness ranking such that it no longer resembles the ranking under complete preferences.
To do so, we evaluate through simulations what percentage of pairwise comparisons of the agents' full rankings are required to be able to compute the divisiveness measure accurately. 

As we are removing parts of the rankings given by the agents, we need to compute divisiveness on incomplete rankings as in the original definition by \citet{navarrete2022understanding}.  When $\score=\copeland$, we see that the definition of divisiveness is well-defined on incomplete rankings. However, on incomplete rankings, we use the win rate instead of $\borda$ when calculating the divisiveness, noting that they are equivalent on complete rankings. We define the win rate of an issue $a$ to be $\sum_{b\in \I\backslash \{a\}} \frac{\#(\N_{a>b})}{\#(\N_{a\succ b }\cup\N_{b\succ a })\cdot (m-1)}$.

\begin{figure}[t!]
    \centering
    \includegraphics[width=0.6\columnwidth]{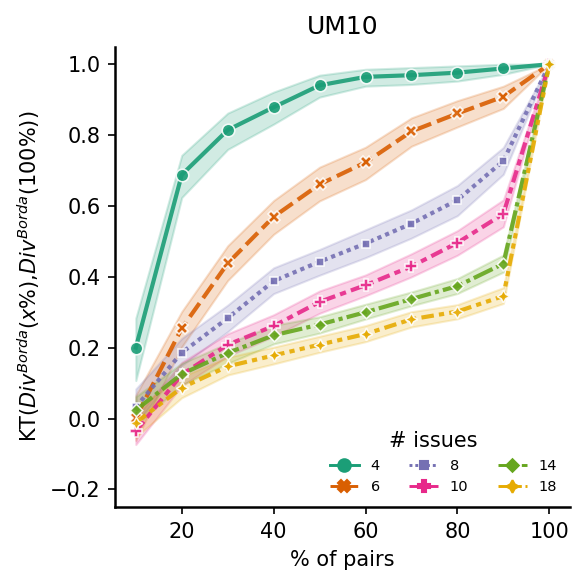}
    \caption{The average Kendall's tau correlation between the divisiveness ranking under $\Div_0^\borda$ computed on the full ranking with respect to incomplete rankings. The horizontal axis shows the percentage of pairwise comparisons of the incomplete profile with respect to the complete one. The experiment is run on 100 preference profiles drawn from UM10. The different markers on the lines represent different numbers of issues in $\{4,6,8,10, 14, 18\}$}
    \label{fig:robust}
\end{figure}

We generated 100 profiles for each of the three preference generation methods IC, UM10, and UM50, varying the number of issues $m\in [3,18]$.
We compared the average Kendall's tau correlation between the divisiveness ranking computed on the full rankings and the adapted measure of divisiveness computed on sub-profiles containing $X\%$ of the pairwise comparisons of the complete one ($X\in[10,100]$ increasing by increments of $10$). Here, pairwise comparisons are deleted from the profile at random. We highlight a single figure from these simulations to illustrate and leave the remainder of the figures in the Appendix. 

The main message of this simulation is that disrupting divisiveness by deleting pairwise comparisons is easy. 
Figure~\ref{fig:robust} focuses on the case when the average KT is taken over $100$ profiles with $100$ agents created via the UM10 method. 
It shows that if there are sufficiently many issues (say more than 10), then deleting between 10 and 20\% of the pairwise comparisons in the profile is sufficient to significantly decrease the accuracy of divisiveness (the correlation between complete and incomplete divisiveness is below $0.5$).
We also observe an inversion of the curves when the number of issues exceeds $6$, tending towards an exponential shape. 
Our findings imply that with a large number of issues and under the assumption of moderately correlated preferences, the almost-totality of the pairwise comparisons needs to be elicited from agents to obtain an accurate measure of divisiveness.

\subsection{Control by Adding Rankings}\label{sec:controladd}

The next form of control we study is the addition of fake rankings by an external agent. 
Modelling this type of control is particularly realistic when the divisiveness measures are used in online forums, where attacks by bots are commonplace (such as in the experimental setting of \citet{navarrete2022understanding}). 
Similar problems were previously studied in the literature, such as Sybil attacks in online elections~\citep{MeirTSS22}. 
We start by presenting an example in which a single agent is able to alter the divisiveness ranking.

\begin{example}\label{ex:manip}
Consider four agents and four issues $\I=\{a,b,c,d\}$. Consider that three agents have submitted their preferences, one agent with $\succ_1$, and two with $\succ_2$. The fourth agent has the truthful preference of $\succ_3$:
 \begin{center}
 \begin{tabular}{l|l}
$\succ_1$ (1 agent) &  $bcad$\\
$\succ_2$ (2 agents) & $abcd$ \\
$\succ_3$ (1 agent) & $acbd$
\end{tabular}  
\end{center}
 If the fourth agent submits their truthful preference, we have that:  $\Div^\borda(a,\P)=\nicefrac{12}{27}$,   $\Div^\borda(b,\P)=\nicefrac{7}{27}$, $\Div^\borda(c,\P)=\nicefrac{8}{27}$, and $\Div^\borda(d,\P)=0$. 
Thus, $a$ is currently the most divisive issue. 
As all agents agree that $d$ is the worst of all the issues, it is the least divisive issue.   

If the agent with truthful preference $\succ_3$ wants to manipulate the divisiveness measure in order to make the issue $b$ more divisive, then they can submit a preference $\succ_3'=cadb$ giving profile $\P'=(\succ_1, (2, \succ_2), \succ_3')$.  
In doing so, the measure of divisiveness changes as such: $\Div^\borda(a,\P')=\nicefrac{19}{54}$,   $\Div^\borda(b,\P')=\nicefrac{38}{54}$, $\Div^\borda(c,\P')=\nicefrac{19}{54}$, and $\Div^\borda(d,\P')=\nicefrac{6}{54}$.
Thus, by submitting $\succ_3'$, the agent succeeded in making $b$ the most divisive issue.
\end{example}

Example~\ref{ex:manip} shows that one agent can manipulate the divisiveness measure. 
We now show that the problem of control by adding rankings can be solved easily by showing a simple heuristic, which we call $\inject{\score}$. 

$\inject{\score}$ takes as input a profile $\P$ and an issue $a$ which it aims to make the most divisive by adding new agents to the profile. It first computes the ranking over issues defined by $\score(x,\P)$, which we denote by $\succ_\score$. 
This is used to create the two rankings which will be added to the profile to increase the divisiveness of $a$, namely $\succ_{odd}$ and $\succ_{even}$. 
The former modifies $\succ_\score$ by putting issue $a$ first and leaving the rest of the ranking unchanged. The latter modifies $\succ_\score$ symmetrically by putting $a$ in the last position and leaving the remaining part of the order unchanged. 
In this way, $\succ_{odd}$ and $\succ_{even}$ resemble the ranking of agreements computed using $\score$, with the only difference being the position of $a$ that alternates between first and last. $\inject{\score}$ then alternates between adding $\succ_{odd}$ and $\succ_{even}$ to profile $\P$ until the issue $a$ is the most divisive.
We show that $\inject{\score}$ always terminates and succeeds in making the target issue the most divisive.

\begin{proposition}
 $\textsc{Inject}_\borda$ always terminates for $\alpha=0$.
\end{proposition}

\begin{proof}
Let $a$ be the target issue. If $a$ is already the most divisive issue in $\P$ then $\textsc{Inject}_\borda$ terminates immediately.
Else, we first prove that  $\Div^\borda(a, \P')>\Div^\borda(a, \P)$ for any profile $\P'$ obtained by using $\inject{\borda}$ on a profile $\P$.
Take an arbitrary profile $\P$, issue $a\in \I$, and $k$ such that $\P'=(\P, (k,\succ_{odd}),(k, \succ_{even}))$, where $\N'$ are the agents in $\P'$. 
Take an arbitrary $b\in \I\backslash \{a\}$,  we have that $\Div^\borda(a,\N_{a\succ b }', \P')>\Div^\borda(a,\N_{a\succ b}, \P)$ as $\borda(a, \P'_{a\succ b})\geq \borda(a, \P_{a\succ b})$ and $\borda(a, \P_{b\succ a})\geq \borda(a, \P'_{b\succ a})$, with at least one of the two inequalities being strict. To see this, observe that all injected agents rank $a$ the highest in $\P'_{a\succ b}$ and the lowest in $\P'_{b\succ a}$ (and $a$ is not the most divisive issue).
Thus, we proved that $\Div^\borda(a,\P')$ tends to 1 as $k$ increases.
To conclude, note that the rank of any issue $b\in\I\backslash\{a\}$ in the injected sub-profile ${((k,\succ_{even}),}(k, \succ_{odd}))$ varies only by one position. Thus, with $k$ large enough the divisiveness of any issue $b\not=a$ cannot tend to $1$.
%
\end{proof}

\begin{proposition}
 $\textsc{Inject}_\copeland$ always terminates in polynomial time for $\alpha=0$.
\end{proposition}

\begin{proof}
Let $k=2n+2$ and $a$ be the target issue, where $n$ is the number of voters in $\P$. By definition, $\Div_0^\copeland(a, \P')=\nicefrac{1}{m-1}\sum_{b\in \I\backslash\{a\}} |\copeland(a, \P'_{a\succ b})-\copeland(a, \P'_{a\succ b})|$. 
As $k$ is sufficiently large, for all $b\in \I\backslash\{a\}$ we have that $a$ is a Condorcet winner in $\P'_{a\succ b}$, since $n+1$ copies of $\succ_{odd}$ were added with $a$ as the top issue. Symmetrically, $a$ is a Condorcet loser in $\P'_{b\succ a}$.
Thus, we have that $\copeland(a, \P_{\N_{a\succ b}})=1$ and $\copeland(a, \P_{\N_{b\succ a}})=0$, which in turn implies that $\Div_0^\copeland(a, \P')=1$. 
By the uniqueness of a Condorcet winner, we have that no other issue $b\not = a$ can have $\Div_0^\copeland(b, \P')=1$, concluding the proof.
\end{proof}

The previous results shows that $\inject{\score}$ can manipulate the divisiveness ranking. However, it does not provide a bound on how many agents $\inject{\borda}$ are required and only provides a large bound for $\inject{\copeland}$, namely $k=2n+2$. 
To complement this, we conducted simulations to estimate how many new agents  $\inject{\score}$ needs to alter the divisiveness ranking. 
For each $m\in[2,11]$ and each of three profile generation methods (IC, UM10, UM50), we considered $100$ profiles to test how many new agents $\inject{\score}$ required to make the target issue the most divisive. 
Figure~\ref{fig:manipbordaICwhathappens} focuses on IC profiles with $8$ issues. It shows the divisiveness rankings of the $8$ issues in the initial profile (at $0\%$) and their evolution when $\inject{\borda}$ inserts additional rankings to make the least divisive issue the most divisive (the highlighted line at the $8^\text{th}$ position).
In particular, by adding around $35\%$ new agents we can make the least divisive issue the most via our simple algorithm. 
In crowdsourcing applications with wide participation this percentage implies that too many additional agents might need to be injected without being noticed.
Yet, we recall that $\inject{\score}$ is a simple heuristic, and this percentage could be lower with a more efficient procedure. 
Furthermore, we see that the average divisiveness ranking of the other issues converges to the middle of the ranking. 
We obtained similar results by varying the number of issues and the profile generation methods (we give more details in the Appendix).

\begin{figure}[t]
    \centering
    \includegraphics[width=0.6\linewidth]{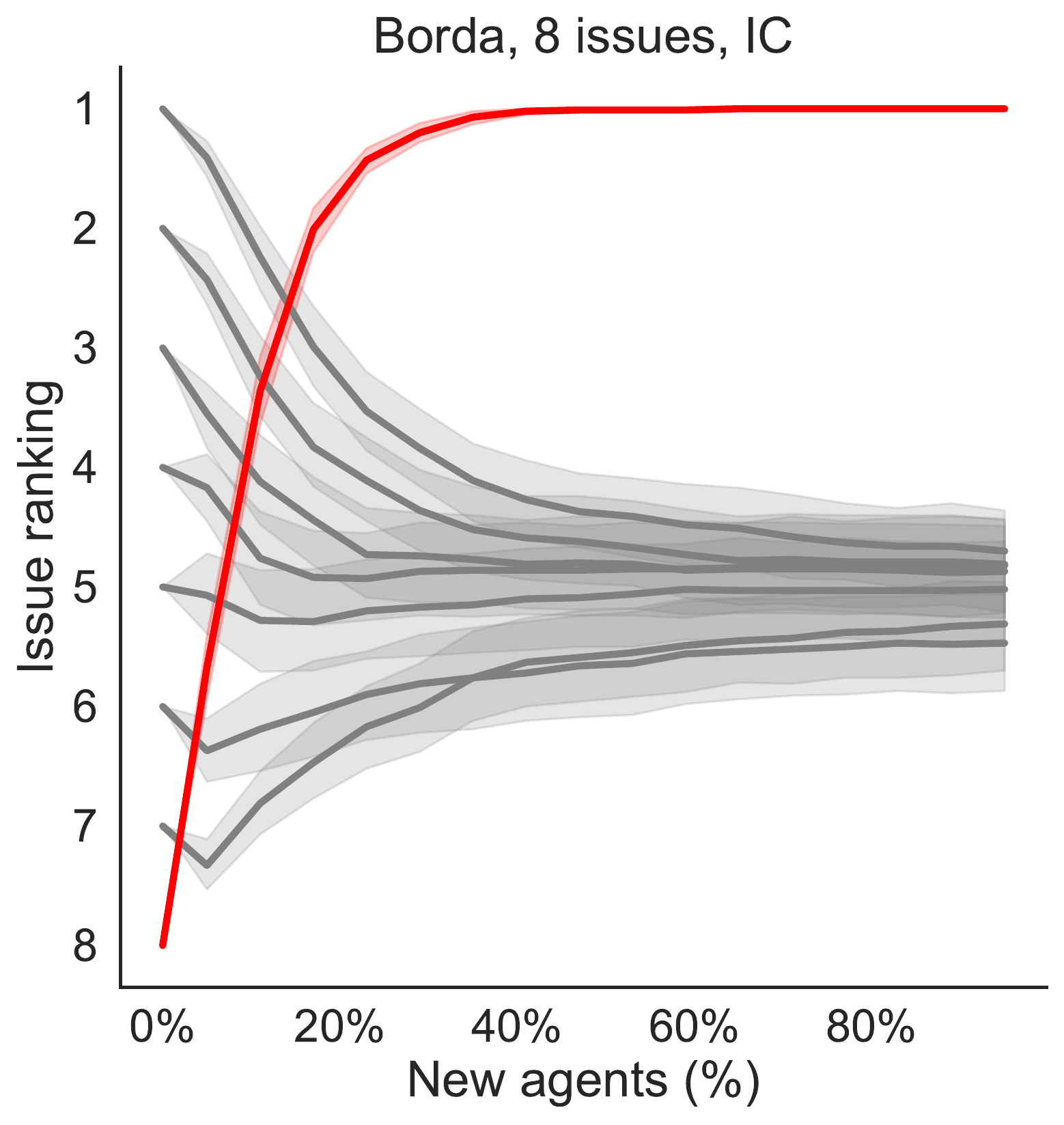}
    \caption{Percentage of new agents added by $\inject{\borda}$ to make the least divisive issue the most divisive. The average divisiveness of the other issues is also plotted in grey. Averages are computed over $100$ IC profiles of $100$ agents and $8$ issues.}
    \label{fig:manipbordaICwhathappens}
\end{figure}

We also tested how many additional agents $\inject{\borda}$ required to reach easier targets, such as making either the second most divisive issue or an issue in the middle of the divisiveness ranking the most divisive issue. 
Figure~\ref{fig:bordaManip} presents our findings for $\score=\borda$ and $m=8$, computed on $100$ profiles generated using either IC or UM50.
Clearly, if the task of control is harder, $\inject{\borda}$ needs additional agents to meet its target. 
More importantly, if we compare the performance of $\inject{\borda}$ on different preference generation methods, we see that more correlated profiles (using UM50 in our case) are harder to control no matter the target issue.

\begin{figure}[t]
    \centering
    \includegraphics[width=0.6\linewidth]{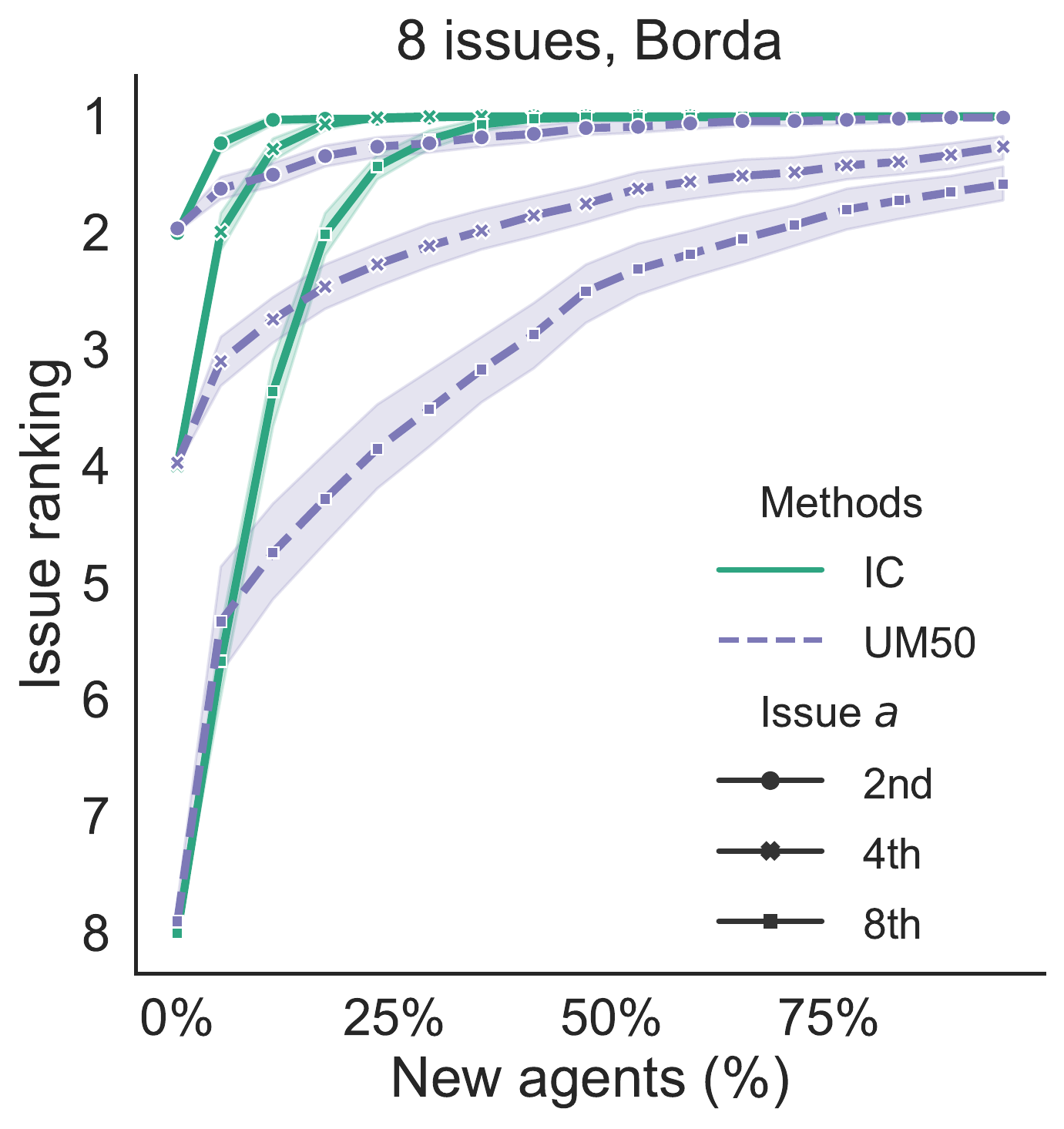}
    \caption{Percentage of agents added by $\inject{\borda}$ with $8$ issues, to make the  $2^\text{nd}$, $4^\text{th}$ or $8^\text{th}$ most divisive issue the most. The average is taken over 100 profiles for 100 agents generated using IC and UM50 (represented by a solid or dashed line, respectively).}
    \label{fig:bordaManip}
\end{figure}

Given the results, we see that $\inject{\score}$ can be an effective way of manipulating, but it simulates a static scenario. 
We also note that this way of manipulating requires very little information about the original profile, i.e., the controlling agent just need to know the current ranking of agreement given by the scoring $\score$. Furthermore, our results show that $\inject{\score}$ can be used just to increase an issue's rank in the divisiveness ranking, rather than insisting that it becomes the most divisive. 
 Additional details can be found in the  Appendix. 

\section{Conclusions and Future Work}

This paper extends the notion of divisiveness given by \cite{navarrete2022understanding} to a family of measures and applies them to complete rankings over issues. 
We ground these measures by highlighting their behaviour at limit cases and comparing them to other notions of disagreement and polarisation.
We also point out how we can find a sub-population for which an issue is most divisive in polynomial time when considering the Borda score, yet, no such algorithm was found for the Copeland score. 
The main contribution of this paper is the study of the robustness of the divisiveness measures to external control. We showed via simulations that by randomly removing pairwise comparisons from the rankings, the correlation between the divisiveness ranking of the full vs partial rankings can drop significantly, especially when there are many issues. 
Furthermore, we show that a simple algorithm can affect the divisiveness ranking by inserting (a possibly large number of) controlled fake rankings.  

This paper opens many directions for future work. First, how can our divisiveness measures be modified to be more robust to external control. Second, our divisiveness measures can be used to compare how divisive or polarised is a given population (instead of focusing on comparisons of single issues). Finally, following in the social choice theory tradition, we will explore axiomatic characterisations of divisiveness.

\bibliographystyle{named}
\bibliography{divisive}

\clearpage

\appendix

\begin{figure*}[t]
    \centering
    \includegraphics[width=0.55\textwidth]{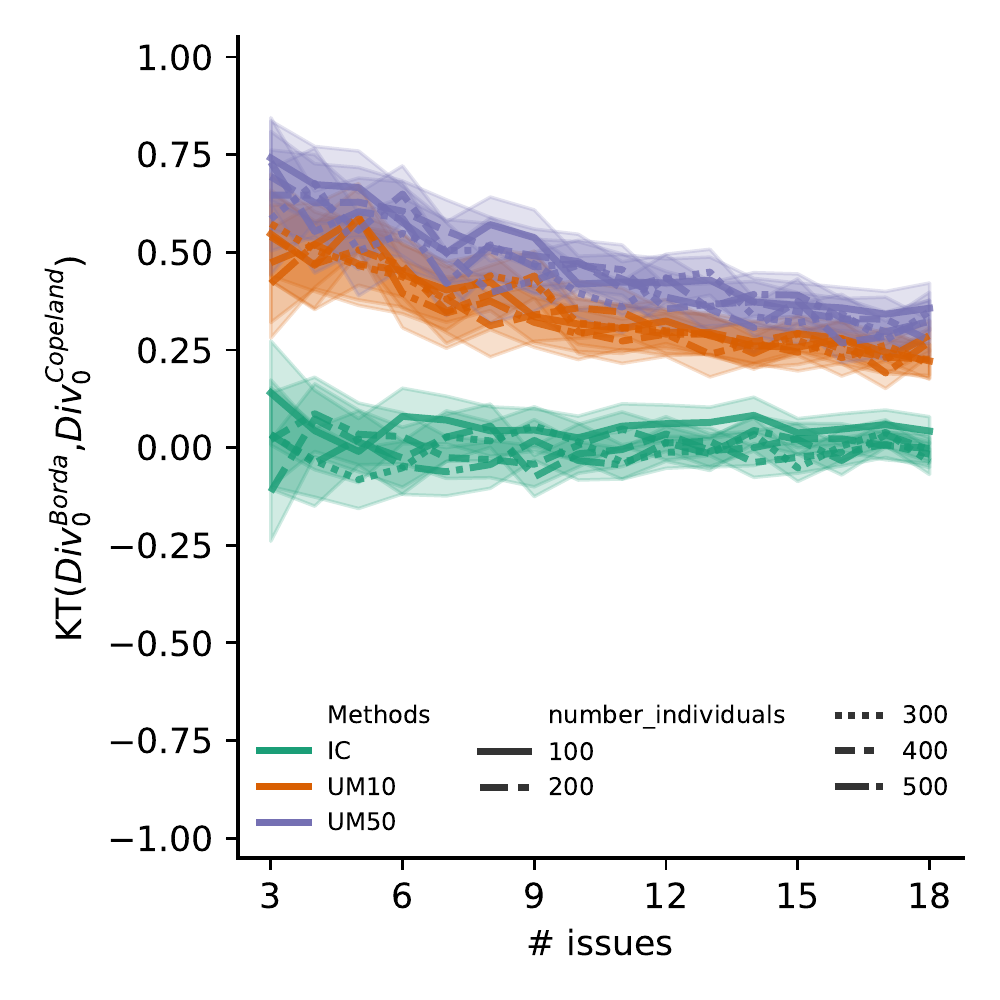}
    \includegraphics[width=0.55\textwidth]{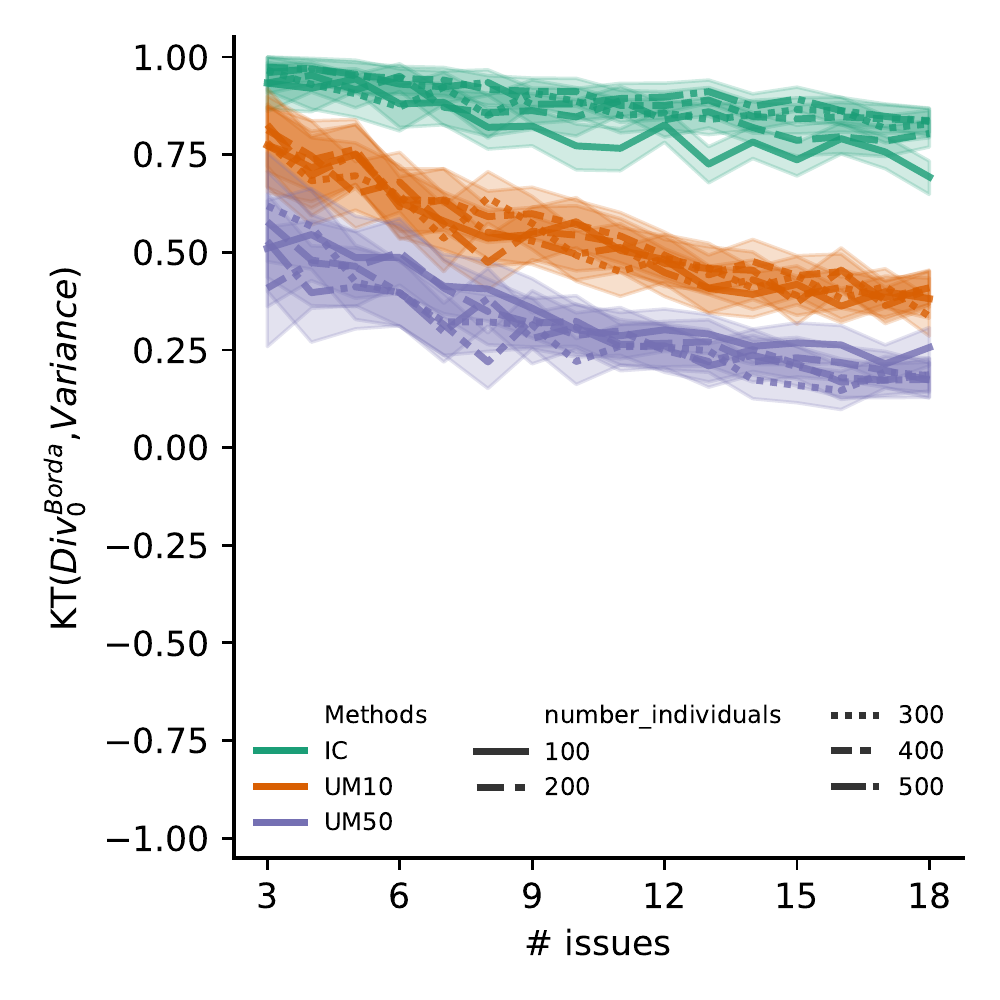}
    \includegraphics[width=0.55\textwidth]{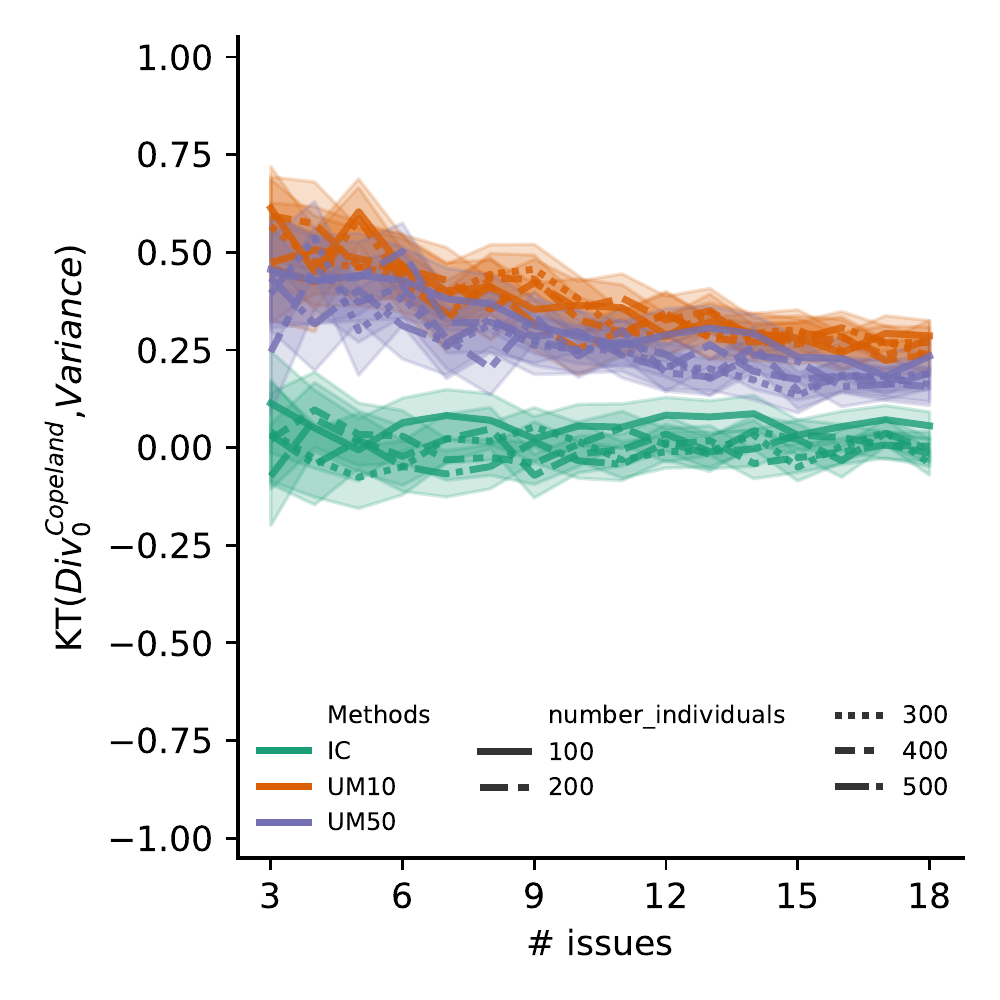}
    \caption{ Kendall's tau correlation between divisiveness using Borda and divisiveness using Copeland; divisiveness using Borda and rank-variance; divisiveness using Copeland and rank-variance (from top to bottom). We measure the correlation against the number of issues the $100$ agents have given a ranking on. The colours represent the methods used to create the rankings, and the style of the lines represents the number of agents $n\in \{100, 200, 300, 400, 500\}$. 
    }
    \label{fig:divvar}
\end{figure*}

\section{Appendix}

In all of our experiments, we built profiles of  $100$ agents, giving a ranking of the issues, the number of which varied from  $3$ to $18$. The profiles were created using three different methods, namely impartial culture (IC) and the urn culture with $10\%$ and $50\%$ correlation (UM10 and UM50) using the tools provided by PrefLib~\cite{MaWa13a,mattei2017apreflib,preflibtools}. 
IC means that when there are $m$ issues,  each of the possible $m!$ rankings are chosen uniformly. The IC method implies that there is not much correlation between the agents; thus, we look at a profile creation method which makes agents more likely to be similar to one another, namely the urn model. We look at the urn model with two values of the similarity parameter, namely being either $10\%$ or $50\%$ correlated. For the urn model, we start with the $m!$ possible rankings an agent could have over the issues. The similarity between agents is determined by this parameter which says how many of these given profiles are returned to the \emph{urn}. Thus to be $50\%$ correlated means that $m!$ of this profile will be returned to the urn, meaning that this ranking will be chosen half of the time when selecting the following ranking. Similarly, to be $10\%$ correlated, the model returns $\nicefrac{m!}{9}$ copies of the ranking. Hence, UM50 provides more correlated profiles, and IC gives us random profiles. Note that when $m>18$, the time to create the profiles via the urn model is difficult due to space and time constraints as $19!\approx 10^{17}$.
To avoid outliers, we take $100$ profiles in our simulations and will average the result for this desired parameter.

\subsection{Additional details of experiments in Section~\ref{sec:divVar}}\label{append:secvardiv}

In this section we give more details of the simulations conducted for Section~\ref{sec:divVar}. These simulations aim to show if there is any correlation between divisiveness (with $\alpha=0$) with either $\score=\borda$ or $\score=\copeland$ and the rank-variance (defined in Section~\ref{sec:variance}). The purpose of this is to test if (i) the choice of scoring function impacts the divisiveness ranking and (ii) if our metrics are, in fact, the rank-variance (a known notion from the social sciences). In the main paper show that the Kendall's tau (KT) correlation between pairs of the three metrics mentioned are positive, yet around a $0.5$ correlation. They differ when the $100$ profiles over $100$ agents are built using the UM10 method.

We give more details about the other methods in Figure~\ref{fig:divvar}. From left to right, we inspect the KT correlation between the ranking of the issues when using the divisiveness issue with either scoring, between divisiveness with the Borda scoring and the rank-variance, and between divisiveness with the Copeland scoring and the rank-variance. In each image, we inspect the effect on the value of the KT correlation when changing the number of issues. Furthermore, we vary the number of agents in the profile with $n\in\{100, 200, 300, 400, 500\}$. However, this seems to have little effect on the KT score. 

The general trend is that the correlation decreases slightly as the number of issues increases. This is to be expected, as when there are more issues, it becomes much harder to have the same ranking as more possible combinations.

One interesting takeaway from these figures is that the order of the lines of which profile creation method at the two measures more correlated changes depending on the variance. We see that $KT(\Div_0^\borda, \Div_0^\copeland)$ and $KT(\Div_0^\copeland, \Var)$ have that they are more correlated when using UM10 and UM50 to create the profiles and less correlated using IC. $KT(\Div_0^\borda, \Var)$ is more correlated with IC than UM10 and UM50.  
In particular, we see that the correlation $KT(\Div_0^\borda, \Var)$ is generally higher than the others. This could be to do with the similar nature of the measures, in that both use the position in the agents' rankings. In contrast, the correlation between either of the other two and Copeland is lower as it looks at which issues are being beaten. Furthermore, we conjecture that on IC profiles  $KT(\Div_0^\borda, \Var)$  is highly correlated as the size of the subgroups in divisiveness are generally equal, and thus, the impact of not using the group sizes is not coming into play.

\begin{figure*}[t]
    \centering
    \includegraphics[width=0.55\linewidth]{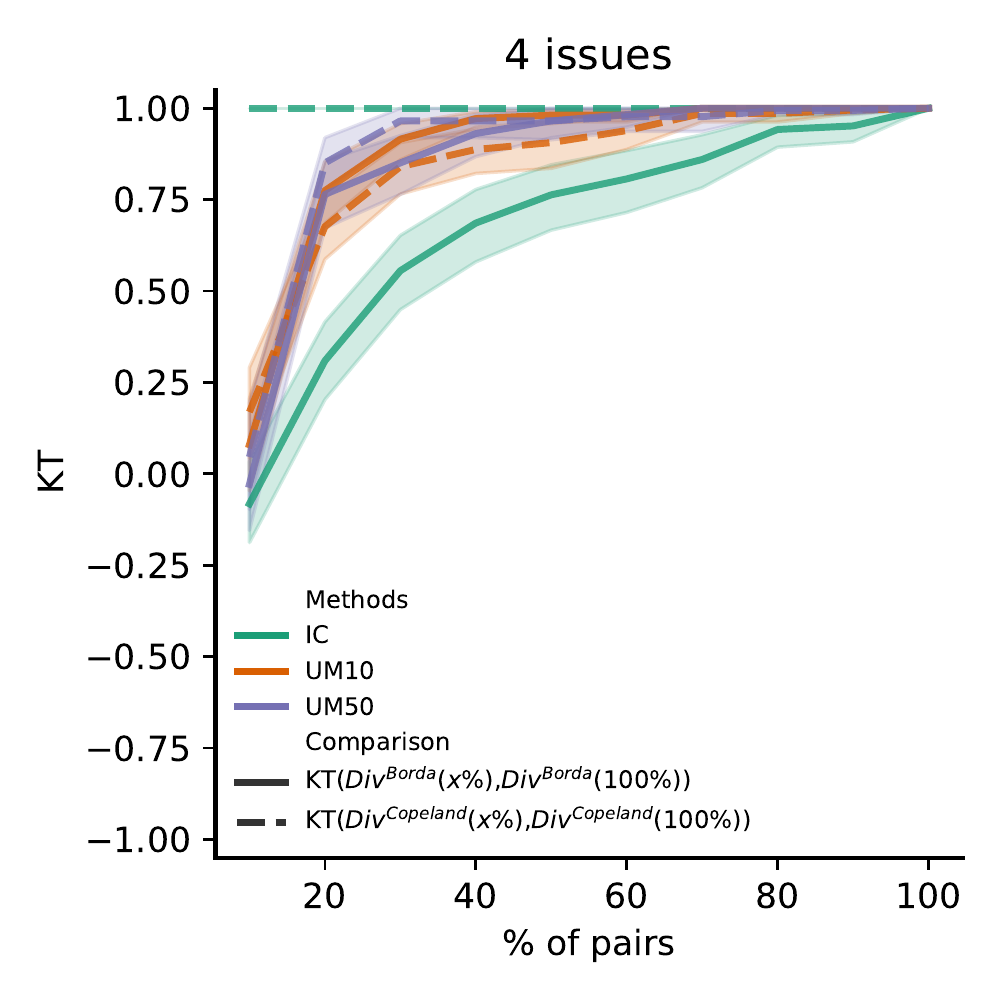}
    \includegraphics[width=0.55\linewidth]{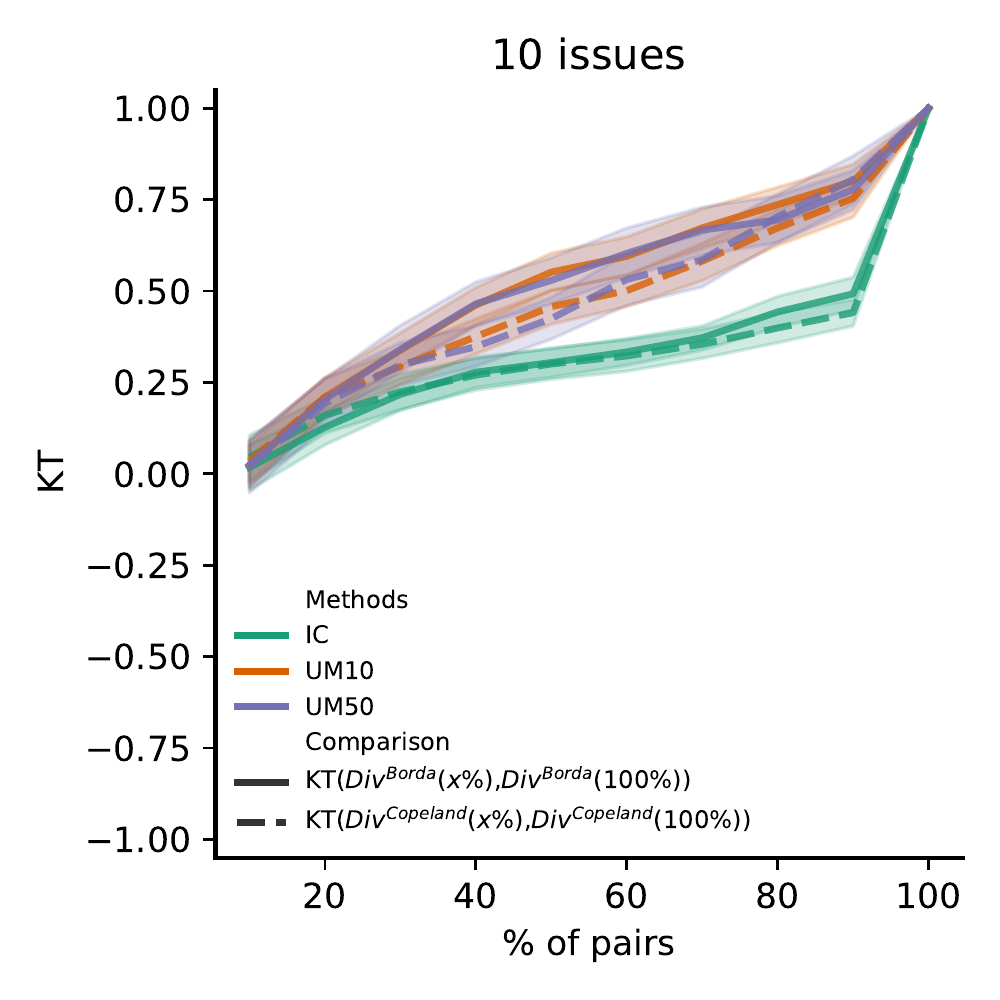}
    \includegraphics[width=0.55\linewidth]{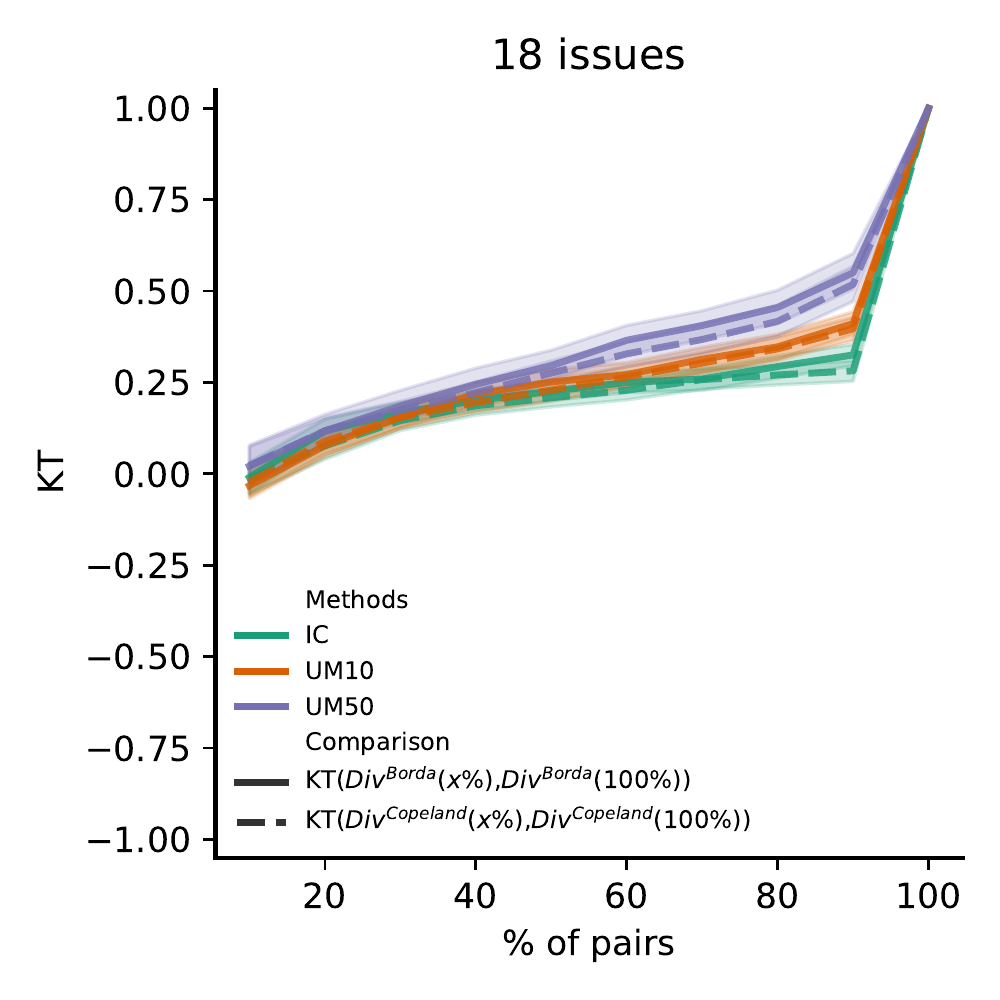}
    \caption{KT correlation between the divisiveness given every agent's full rankings are known and when only a certain percentage of the full preferences are known computed using both the Borda or Copeland scoring (represented by the solid and dashed line, respectively). The colours represent the methods to create the rankings, IC, UM10, UM50.}
    \label{fig:robustAppendix}
\end{figure*}


\subsection{Additional details of the experiments in Section~\ref{sec:robustness}}\label{app:sec:robustness}

The robustness of rank aggregators has been widely studied, either in terms of communication complexity, i.e., bounding how much information needs to be elicited to compute the full ranking \cite{Conitzer2005communication}, or, more recently, by assessing the effect of perturbations of the input \cite{kahng2019statistical}.
In this section we evaluate, via simulations, how many queries on pairwise comparisons are needed from voters to compute a robust notion of divisiveness (assuming $\alpha=0$, for $\score=\borda$ and $\score=\copeland$). 
The pairwise comparisons from the agents' full rankings are removed at random. 
Here, we focus on $n=100$ and $100$ profiles of each type and take the average over each of the correlations.

As we are removing parts of the rankings given by the agents, we need to compute divisiveness on incomplete rankings as in the original definition by \citet{navarrete2022understanding}.  When $\score=\copeland$, we see that the definition of divisiveness is well-defined on incomplete rankings. However, on incomplete rankings we use the win rate when calculating the divisiveness instead of  $\borda$, noting that they are equivalent on complete rankings. We define the win rate of an issue $a$ to be $\sum_{b\in \I\backslash \{a\}} \frac{\#(\N_{a>b})}{\#(\N_{a\succ b }\cup\N_{b\succ a })\cdot (m-1)}$.

In Figure~\ref{fig:robustAppendix}, we show the impact of removing pairwise comparisons from the agents' rankings on how accurate the divisiveness measure is on a partial profile. 
We focus on $m\in \{4,10, 18\}$, shown in the figures from left to right. 
As a general trend, when there are more issues, we need more information about the full rankings for the divisiveness ranking to be more accurate. For the average KT score to be at least $0.5$, we need roughly $20\%$, $70\%$, and $90\%$ of the ranking when the numbers of issues are $4, 10$, and $18$, respectively. 
Moreover, we see that when the profiles are less correlated, i.e., created with the IC method, more pairwise comparisons are required than the other methods to be as accurate.
Lastly, the choice of score does not make much of a difference.

One takeaway from these simulations is that for many issues, it is hard to predict the divisiveness ranking under incomplete information accurately. Thus, there can be an element of control, as if the agents are restricted to answering a subset of pairwise comparisons, then it may be possible to make one issue divisive easily.


\begin{figure*}[t]
    \centering
    \includegraphics[width=\textwidth]{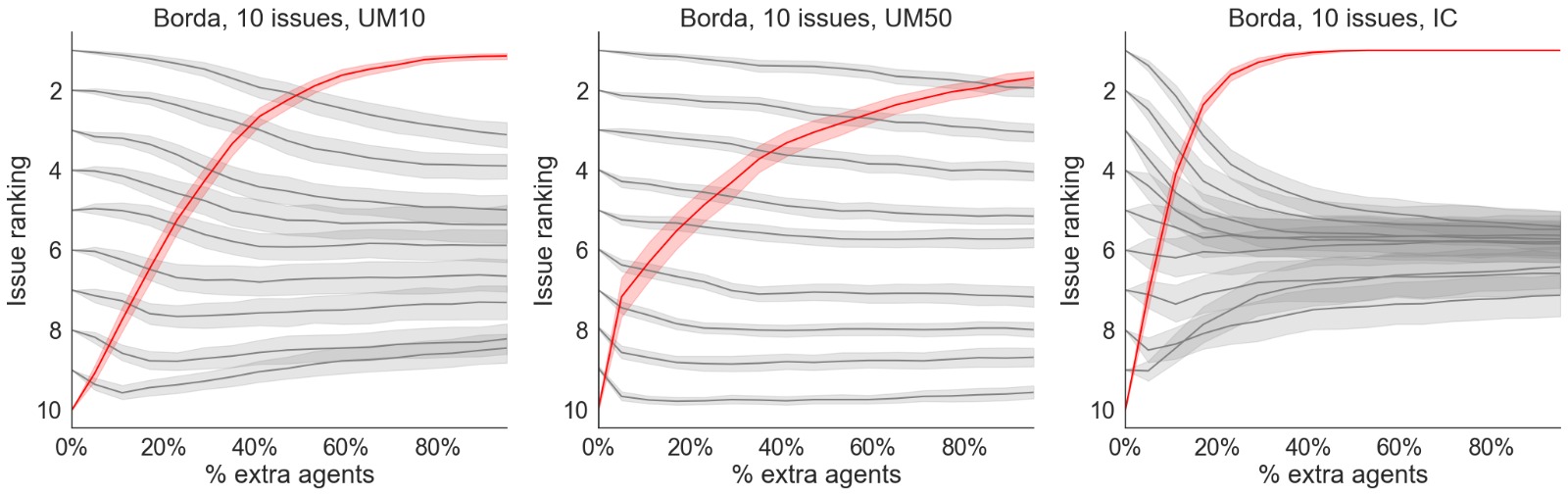}
    \caption{The effect of the heuristic $\inject{\borda}$ on $100$ profiles built using the method UM10, UM50 and IC, respectively from left to right with $50$ agents and $10$ issues. The red line highlights the target issue, the last issue in the divisiveness ranking, becoming more divisive via the addition of new agents via $\inject{\borda}$.}
    \label{fig:exp3followordering}
\end{figure*}

\begin{figure*}[t!]
    \centering
    \includegraphics[width=\textwidth]{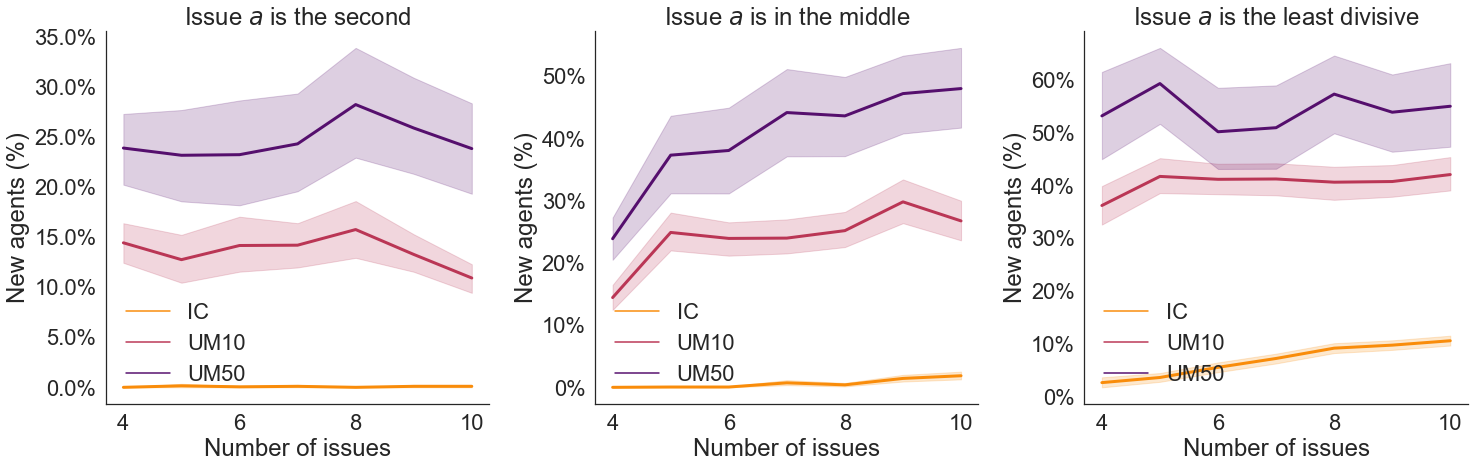}
    \caption{This figure shows the effect of the number of issues for each of the different profile creation methods on the number of new agents required by $\inject{\copeland}$ to make the target issue $a$ the most divisive issue. From left to right, $a$ begins, in the original profile, being the second most divisive issue, middle of the divisiveness ranking, and the least divisive issue.  }
    \label{fig:exp3percentperperissue}
\end{figure*}

\subsection{Additional details of experiments in Section~\ref{sec:controladd}}\label{app:sec:manip}

In this section we inspect the simulations to complement the results in Section~\ref{sec:controladd}, where we showed that a simple heuristic (described in Section~\ref{sec:controladd}) can change the divisiveness ranking such that a target issue becomes the most divisive in the ranking by adding in extra controlled agents, which can be thought of as bots or Sybil agents. These simulations reflect the robustness of the divisiveness measures when used in online contexts where the creation of additional agents is possible.

In Figure~\ref{fig:exp3followordering} we show the effect of the heuristic on $100$ profiles for each creation method (IC, UM10, and UM50) with $10$ issues and $50$ agents in the original profile using $\score=\borda$ and $\alpha=0$. The purpose of these images is to show the effect on all of the issues when the target issue originally was the least divisive. 

For each plot in Figure~\ref{fig:exp3followordering}, we see that the divisiveness of the non-target issues slightly tends to the middle rankings but generally remains roughly the same (each of the grey lines). Next, we see that the more correlated the agents are, the longer it takes the final issue in the divisiveness ranking to become the most divisive issue using this heuristic. Thus, the IC profiles have the target issue becoming the most divisive much more quickly than the UM10 profiles, which in turn has the least divisive issue becoming the most, more rapidly than in the UM50 profiles. Note that when computing the same when $\score=\copeland$, that the images are very similar.

To complement Figure~\ref{fig:bordaManip}
We now inspect the effect of the number of issues on 
$n=100$ and $100$ profiles for each $m\in \{4,\cdots,  10\}$ and each profile creation method IC, UM10, UM50 (corresponding to the yellow, red, and purple lines, respectively, in the figure). In general, the number of additional agents required is affected much by the number of issues in the profile. An intuitive aspect seen in the figure is that the fewer positions required for $\inject{\copeland}$ to move the target issue, the fewer extra agents are needed, i.e., more agents are required to make the least divisive issue the most, rather than the second most divisive issue the most.   

A significant takeaway from the simulations seen in the figures is that profiles created with the IC method are very easily manipulated by $\inject{\copeland}$,  with even $10$ issues and the target issue being initially the least divisive issue that only $10\%$ extra agents are required. 
With more correlated profiles, i.e., made with UM10 and UM50, we see that more agents are required to manipulate the divisiveness ranking. We see that around an additional $50\%$ of the agents are required to make the least divisive issue the most in significantly correlated profiles. Hence, the heuristic is an effective way of manipulating the divisiveness ranking. Furthermore, for more modest aims, such as making the second most divisive issue the most, it seems possible by adding $15\%$ new agents for UM10 and $25\%$ for UM50. 

The same experiment for $\score=\borda$ produces similar figures, yet with slightly more noise. 

\end{document}